\documentclass[conference]{IEEEtran}%
\IEEEoverridecommandlockouts
\usepackage{amsfonts}
\usepackage{amsmath}
\usepackage{amssymb}
\usepackage{amsthm}
\usepackage{balance}
\usepackage{cite}
\usepackage[]{algorithm2e}
\usepackage{subfigure}
\PassOptionsToPackage{hyphens}{url}
\usepackage{url}

\usepackage[pdftex]{graphicx}
\usepackage{epstopdf}
\graphicspath{{./}}
\usepackage{multirow}

\usepackage{pifont}

\newtheorem{lemma}{Lemma}
\newtheorem{theorem}{Theorem}
\newcommand{\pb}{p_{\mathsf{b}}}

\newcommand{\alphaL}{\alpha_{\mathsf{L}}}
\newcommand{\alphaN}{\alpha_{\mathsf{N}}}

\newcommand{\ISPhi}{I_{\mathbf{\Phi}}^{\mathsf{SI}}}
\newcommand{\IWPhi}{I_{\mathbf{\Phi}}^{\mathsf{WI}}}
\newcommand{\BPhi}{\mathbf{\Phi}}
\newcommand{\RB}{R_{\mathsf{B}}}

\newcommand{\zR}{z_{\mathsf{R}}}
\newcommand{\xR}{x_{\mathsf{R}}}
\newcommand{\yR}{y_{\mathsf{R}}}
\newcommand{\BR}{B_{\mathsf{R}}}
\newcommand{\jj}{\mathrm{j}}
\newcommand{\BL}{\mathsf{B}_{\mathsf{L}}}
\newcommand{\ddh}{d_{\mathsf{h}}}

\newcommand{\be}{\begin{eqnarray}}
\newcommand{\ee}{\end{eqnarray}}

\usepackage[usenames]{color}

\begin{document}
\hyphenation{multi-symbol}
\title{Location Based Performance Model for Indoor mmWave Wearable Communication}
\author{ Kiran Venugopal and Robert W. Heath, Jr.\\
\thanks{Kiran Venugopal and Robert W. Heath, Jr. are with the University of Texas, Austin, TX, USA. This work was supported in part by the Intel-Verizon 5G research program.}
}
\date{}
\maketitle

\vspace{-1cm}
\thispagestyle{empty}
\begin{abstract}
Simultaneous use of high-end wearable wireless devices like smart glasses is challenging in a dense indoor environment due to the high nature of interference. 
In this scenario, the millimeter wave (mmWave) band offers promising potential for achieving gigabits per second throughput. Here we propose a novel system model for analyzing system performance of mmWave based communication among wearables. The proposed model accounts for the non-isotropy of the indoor environment and the effects of reflections that are predominant for indoor mmWave signals. The effect of human body blockages are modeled and the system performance is shown to hugely vary depending on the user location, body orientation and the density of the network. Closed form expressions for spatially averaged signal to interference plus noise ratio distribution are also derived as a function of the location and orientation of a reference user.
\end{abstract}

\section{Introduction}

Research predicts that in five years, users will have three to eight electronic devices with different capabilities on and around their bodies \cite{IDC_forecast, Winter:2015}. These devices can work more efficiently if they are networked together to form what is called a wearable network.  For example, devices may be cheaper and have lower power requirements if they use a short-range connection to the smart-phone, instead of their own dedicated cellular connection. A major challenge in wearable networks is ensuring their efficient operation when there is a high density of users, as for example in public transport during rush hour. 

Since millimeter (mmWave) frequency bands have large bandwidths and good isolation features, they serve as ideal candidates to deliver high data rates in wearable networks. This is especially promising for dense indoor deployments in train cars and buses. To quantify the performance gains, analysis of mmWave based wearable network was reported in \cite{mmWave:2015} using the approach in \cite{torrieri:2012} that originally deals with the outage performance of finite ad-hoc networks in the non-mmWave frequency setup. In \cite{mmWave:2015, mmWave_arXiv:2015}, human bodies were modeled as the main source of signal blockage for mmWave based wearable networks as against buildings in the outdoor mmWave cellular setup in \cite{bai:2014}. Using the approaches in \cite{torrieri:2012,bai:2014}, spatially averaged closed form expression for the signal to interference plus noise ratio (SINR) coverage probability for a typical user located at the center of a dense but finite network region was derived in \cite{mmWave_arXiv:2015}. The main limitation in \cite{mmWave:2015, mmWave_arXiv:2015} is that the effect of wall and ceiling reflections for the metallic indoor environment was not explicitly modeled. The explicit effects of first order reflections from all the six faces of a cuboidal enclosure were considered for the simulation results in \cite{Geordie:mmWave}. This provided valuable insights about the nature of surface reflections in the indoor mmWave setup. While it was assumed that the reflections from the ceilings were never blocked and the self-body human blockage was characterized, \cite{Geordie:mmWave} does not report closed form analytic expressions for spatially averaged SINR performance of a typical user.

In this paper, we propose a system model that admits easy analysis while considering the significant effects of wall and ceiling reflections. We assume the users are drawn from a region of a Poisson Point Process (PPP) in ${\mathbb{R}}^2$  and derive closed form expressions for the spatially averaged SINR coverage probability of a reference user as a function of location within the enclosure. The effects of body orientation of the reference user and the density of the users are also characterized.

\begin{figure}
\centering
\includegraphics[totalheight=2.2in,width=3.3in]{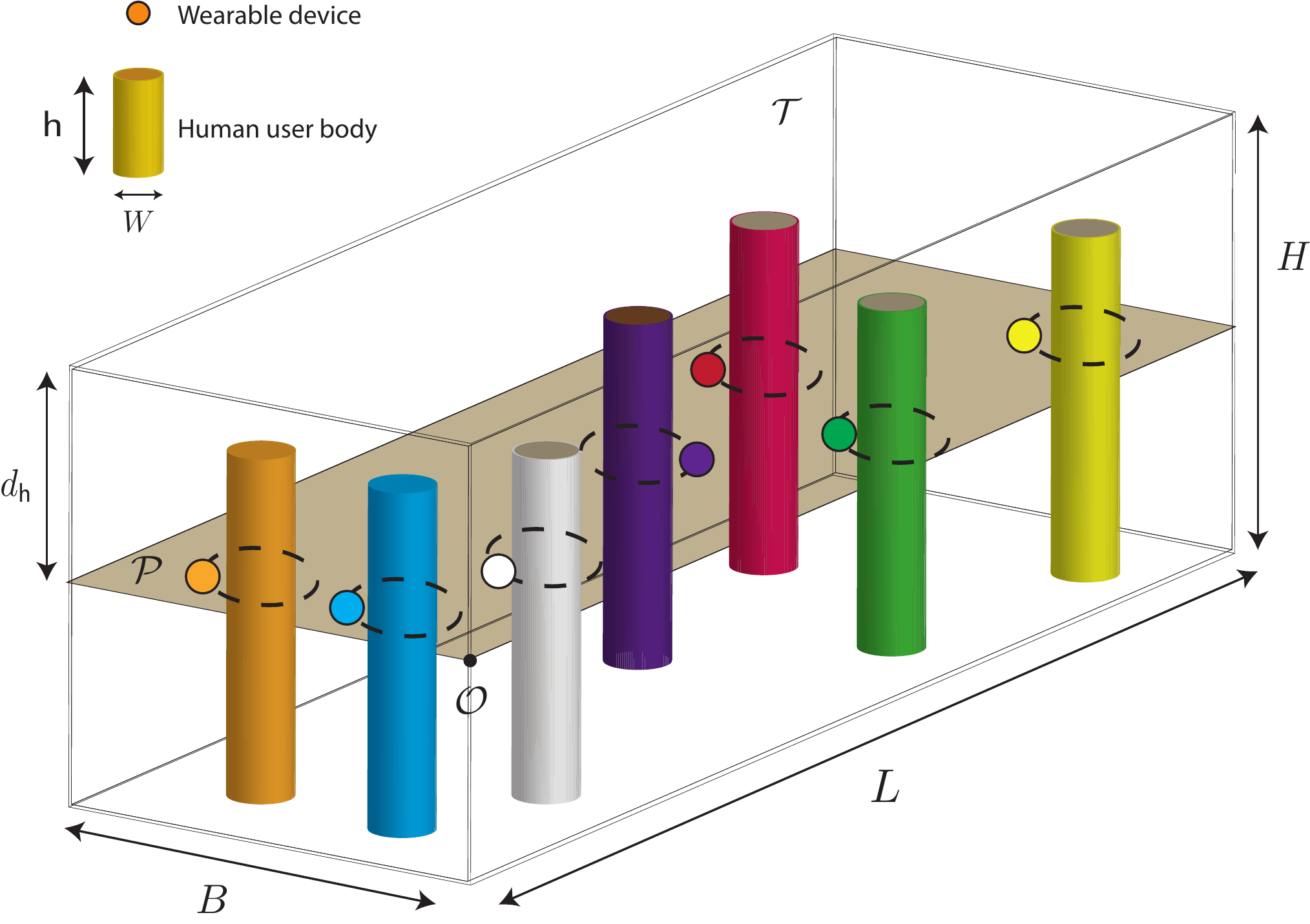}
\caption{Enclosed region $\mathcal{T}$ showing the random locations of the human users and their corresponding wearable devices positioned in the plane $\mathcal{P}$.}     
\label{fig:train_3D}        
\end{figure}

\section{Proposed System Model}
In this section, we explain the network geometry, blockage model and the propagation features assumed in the paper. Intuitive reasoning for the assumptions made are also provided at relevant places.

\subsection{Network Model}
We consider an enclosed space $\mathcal{T}$ of dimensions $L \times B \times H$ as shown in Fig. \ref{fig:train_3D} which has highly reflective walls and ceiling, and non-reflective floor. The users are assumed to be distributed randomly within $\mathcal{T}$ and are modeled as cylinders ${\mathcal{C}}_{\mathsf{u}}$ of a fixed diameter $W$. Each user is assumed to be equipped with one high-end wearable receiver and one controlling hub (smartphone) which acts as the device-to-device communication transmitter.

 All the transceiver devices are assumed to be positioned at a depth $\ddh$ from the ceiling of $\mathcal{T}$ along the plane denoted as $\mathcal{P}$ (Fig. \ref{fig:train_3D}) and positioned randomly on a circle of radius $d \geq W/2$ concentric with their associated user ${\mathcal{C}}_{\mathsf{u}}$. The locations of the users (centers of ${\mathcal{C}}_{\mathsf{u}}$ intersecting $\mathcal{P}$) are assumed to be drawn from a non-homogeneous PPP $\BPhi$ that has intensity $\lambda$ in the region of interest, and zero outside it. The reference receiver and its location are denoted by the complex number $\zR = \xR+\jj\yR$, where $\xR$ and $\yR$ are the real and imaginary parts, respectively. Similarly, the interfering transmitters and their locations are denoted as $z_i = x_i +\jj y_i$. In this representation, the point $\mathcal{O}$ as shown in Fig. \ref{fig:train_3D}, is assumed to be the origin. 

\subsection{Signal Model}
We assume the power gain $h_i$ due to fading for the wireless link from $z_i$ to the reference receiver is independent and identically distributed normalized Gamma random variable with parameter $m$. The reference signal link is assumed to be of a fixed length $d_0$ with a path-loss exponent of $\alphaL$, where the subscript $\mathsf{L}$ denotes line-of-sight (LOS). The non-LOS (NLOS) path-loss exponent is denoted as $\alphaN$, the relevance of which is explained momentarily. The fade gain $h_0$ of the reference link is also assumed to be a normalized Gamma distributed random variable with parameter $m$. The path-loss function $\ell(\zR,z_i)$ for the link from $z_i$ to the reference receiver depends on the relative position of $z_i$ with respect to $\zR$. The transmitters and the receivers are assumed to be equipped with omni-directional antennas, transmitting with a constant power $P$. The noise power normalized by the signal power observed at a reference distance is denoted as $\sigma^2$. 

\subsection{Modeling Interference and Blockages}
Human bodies are the major source of blockages for indoor mmWave communication \cite{Lu:ZTE, Bai:Asilomar14}. As in prior work \cite{mmWave:2015, Geordie:mmWave, mmWave_arXiv:2015}, we model the body blockage of a human as a diameter $W$ disk located in $\mathcal{P}$. The system model including blockage is shown in Fig. \ref{fig:train_2D}. The diameter-W disk associated with interferer $z_i$ and its location in $\mathcal{P}$ are denoted by the complex number $B_i$. Similarly, the body blockage associated with the reference user and its location are denoted as $\BR$. Further, user $i$ is assumed to be facing towards a direction $\psi_i$ relative to $z_i$, i.e., $\measuredangle \left( z_i - B_i\right) = \psi_i$. The reference receiver $\zR$ is assumed to be facing towards the direction $\psi_{\mathsf{R}}$. 
\begin{figure}
\centering
\includegraphics[totalheight=1.8in,width=3.3in]{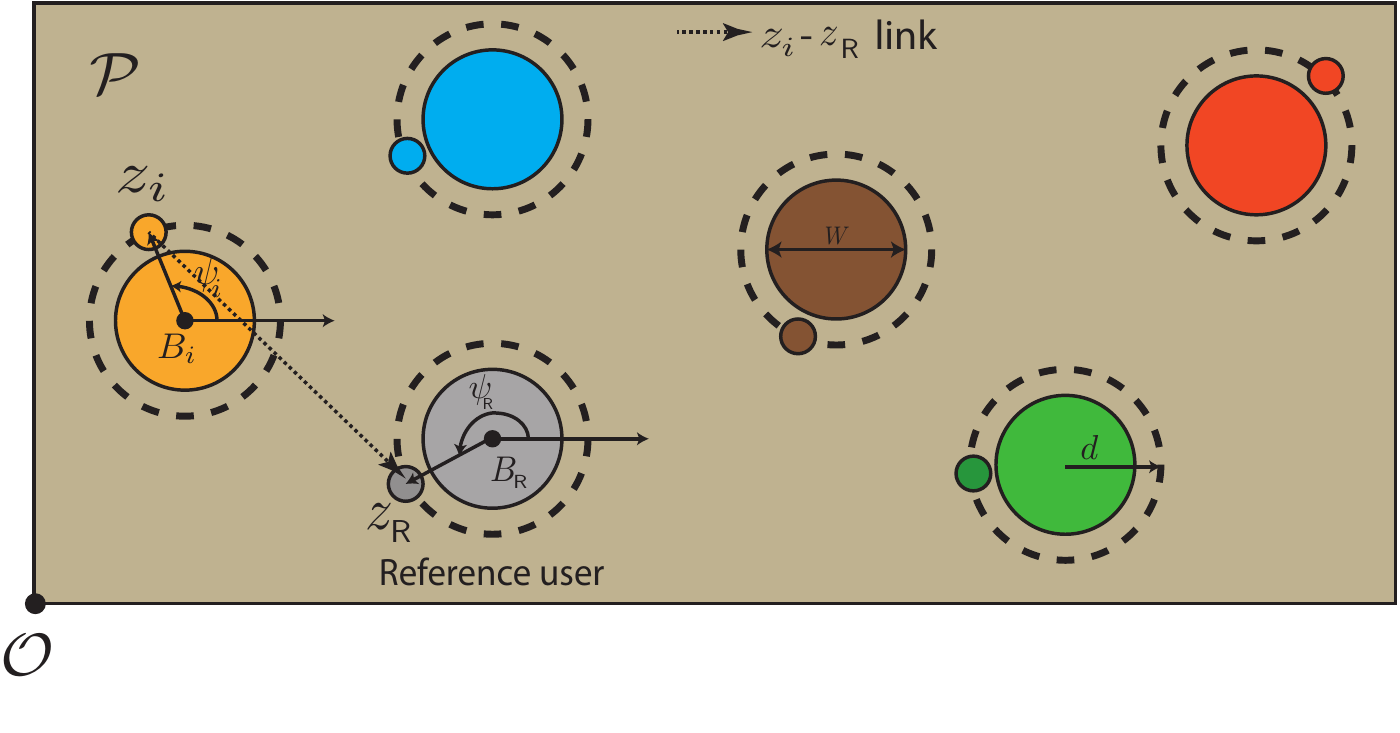}
\caption{2-D abstraction in $\mathcal{P}$ for modeling human body-blockage and user-body orientation representation used in this paper. Wearable devices are randomly located at a distance $d$ around the diameter-$W$ user body.}     
\label{fig:train_2D}        
\end{figure} The signal from an interferer $z_i$ to $\zR$ can be potentially blocked by $B_i$, $\BR$ and/or $B_j,~j\neq i$. The blockage by user $i$ and the reference (termed self-blockage in \cite{Bai:Asilomar14}) can occur irrespective of the locations of $z_i$ and $\zR$. Specifically, self-blockage depends on whether $z_i$ and $\zR$ are facing each other or not. The blockage by user $j \neq i$, on the other hand, depends on the relative separation between $z_i$ and $\zR$, and their individual positions with respect to the reflecting walls of the enclosure. For this reason, self-blockage and blockage by user $j \neq i$ are treated separately. 

\subsubsection{Blockage of $z_i$'s signal by user $j \neq i$}
To see if $B_j, ~j \neq i$ blocks $z_i$, we use the approach in \cite{mmWave:2015} and define a region $\mathcal{BC}(B_j) \in \mathcal{P}$ for each user relative to $\zR$,
\be
\nonumber \mathcal{BC}(B_j) \hspace{-0.1in}&=& \hspace{-0.1in}\left\lbrace z \in \mathcal{P} : | z - \zR|^2 \geq |B_j-\zR|^2 - \left(\frac{W}{2}\right)^2, \right. \\ 
|\measuredangle ( z \hspace{-0.1in}&&\hspace{-0.25in} \left.- \zR) - \measuredangle ( B_j - \zR) | \leq \sin^{-1} \left(\frac{W}{2|B_j-\zR|}\right)\right\rbrace, \label{Equation:blocking_cone}
\ee referred to as the \textit{blocking cone} of $B_j$. Note that, using \eqref{Equation:blocking_cone}, we can also define the blocking cone of the reference user which is denoted as $\mathcal{BC}(\BR)$. The concept of blocking cone is illustrated in Fig. \ref{fig:bc_zi}.
\begin{figure}
\centering
\includegraphics[totalheight=1.8in,width=3.3in]{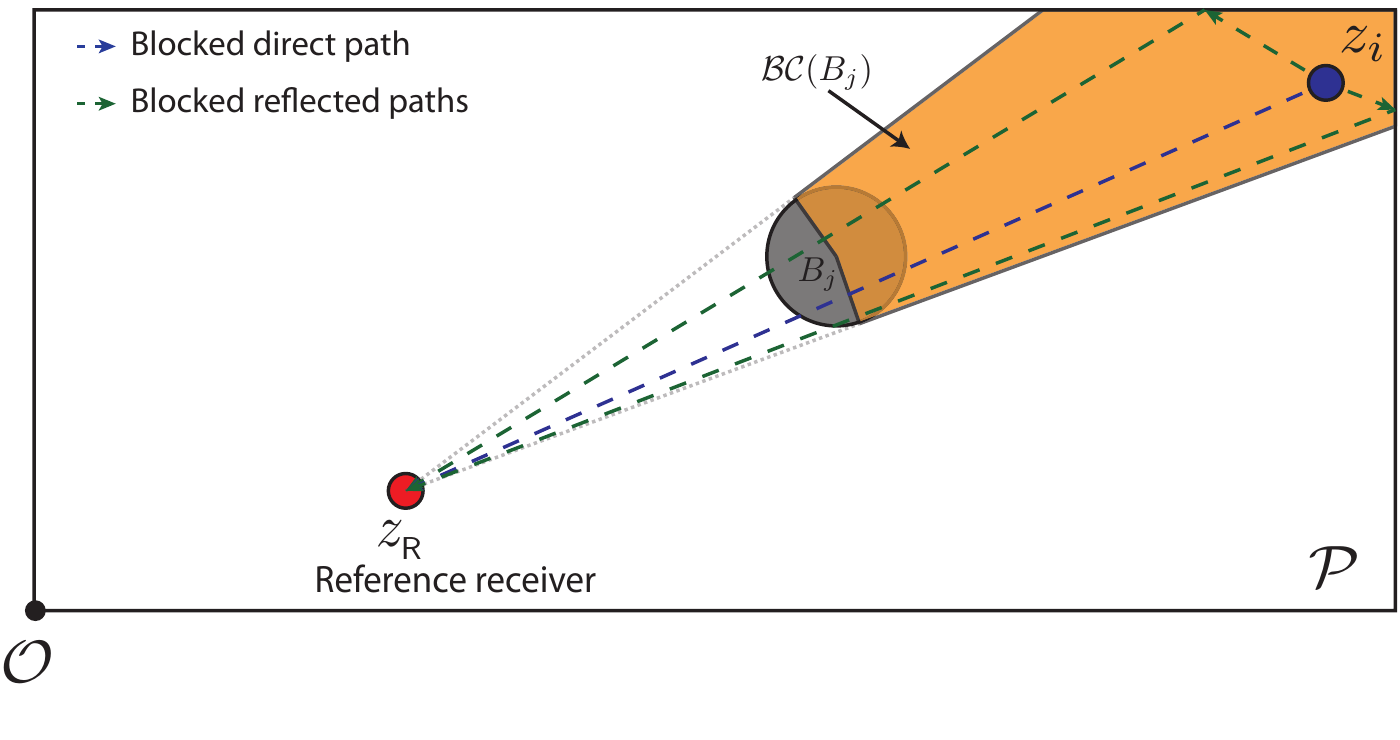}
\caption{Illustration of a blocking cone showing the direct and wall-reflected signal paths from $z_i$ getting blocked by $B_j$.}     
\label{fig:bc_zi}        
\end{figure}
We denote $z_i$ as \textit{strong} if its direct and wall reflected paths are not blocked by any user $j\neq i$. Otherwise, $z_i$ is denoted as a \textit{weak} interferer. Since $B_j$ is uniformly distributed in $\mathcal{P}$, the interfering transmitters located farther away from $\zR$ have a higher chance of being a weak interferer. For analytic tractability, we define a threshold distance $\RB(\zR)$ from the reference receiver such that if $|\zR-z_i| \leq \RB(\zR)$, $z_i$ is a strong interferer. Having $|\zR-z_i| > \RB(\zR)$ implies that there always exists some user $j \neq i$ that blocks the direct and wall-reflected propagation paths from $z_i$ to $\zR$. This threshold distance based model captures the blockage effects due to a third user $j$ for the interference signal from user $i$.

By definition, $z_i$ is a strong interferer whenever $|\zR-z_i| \leq \RB(\zR)$, and there exists no $B_j, j\neq i$ in the path from $z_i$ to $\zR$. Fig. \ref{fig:blocking_zone_center} shows an illustration of the blocking region $\mathcal{A}(\zR,z_i)$ that is used to check if $z_i$ is blocked by a user $j \neq i$, i.e., whenever $B_j \in \mathcal{A}(\zR,z_i)$, $z_i$ is blocked from $\zR$. Since the users are assumed to be drawn from $\BPhi$, the probability that there is no user in the region $\mathcal{A}(\zR,z_i)$ is $\exp(-\lambda |\mathcal{A}(\zR,z_i)|)$, where $|\mathcal{A}(\zR,z_i)|$ is the area of $\mathcal{A}(\zR,z_i)$. The shape of $\mathcal{A}(\zR,z_i)$ varies with $\zR$ (and $z_i$). In particular, the variation in the shape and hence the area $|\mathcal{A}(\zR,z_i)|$ is more pronounced when one or both of $\zR$ and $z_i$ are near the walls as shown in Fig. \ref{fig:blocking_zone_edges}. The effect of wall-reflections - which results in a near LOS signal propagation \cite{Geordie:mmWave} - is, however, higher when the receiver and/or the interfering transmitter are closer to the wall (Fig. \ref{fig:wall_reflections}). In a densely crowded environment, since the reflected interference signals need to propagate through a longer path, the probability that the onward and reflected paths for the interference bouncing off a wall are not blocked is higher. This means a different region $\mathcal{A}'(\zR,z_i)$ needs to be considered as illustrated in Fig. \ref{fig:blocking_zone_actual}.

\begin{figure}
\centering
\subfigure[Case when $z_i$ and $\zR$ are away from the reflecting walls. The wall-reflected paths have lengths far larger than the direct path length and hence are not considered.]{
\includegraphics[totalheight=1.6in,width=3.3in]{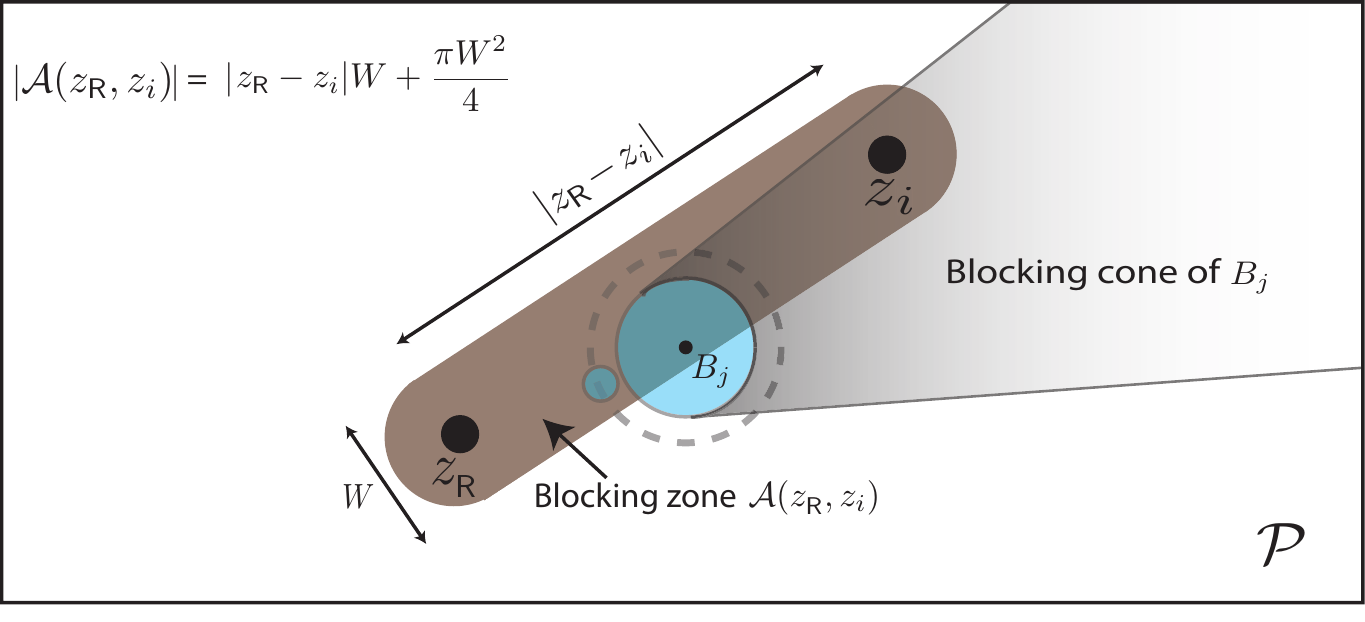}    
\label{fig:blocking_zone_center}        
}
\subfigure[Case when $z_i$ and $\zR$ are near the reflecting walls.]{
\includegraphics[totalheight=1.6in,width=3.3in]{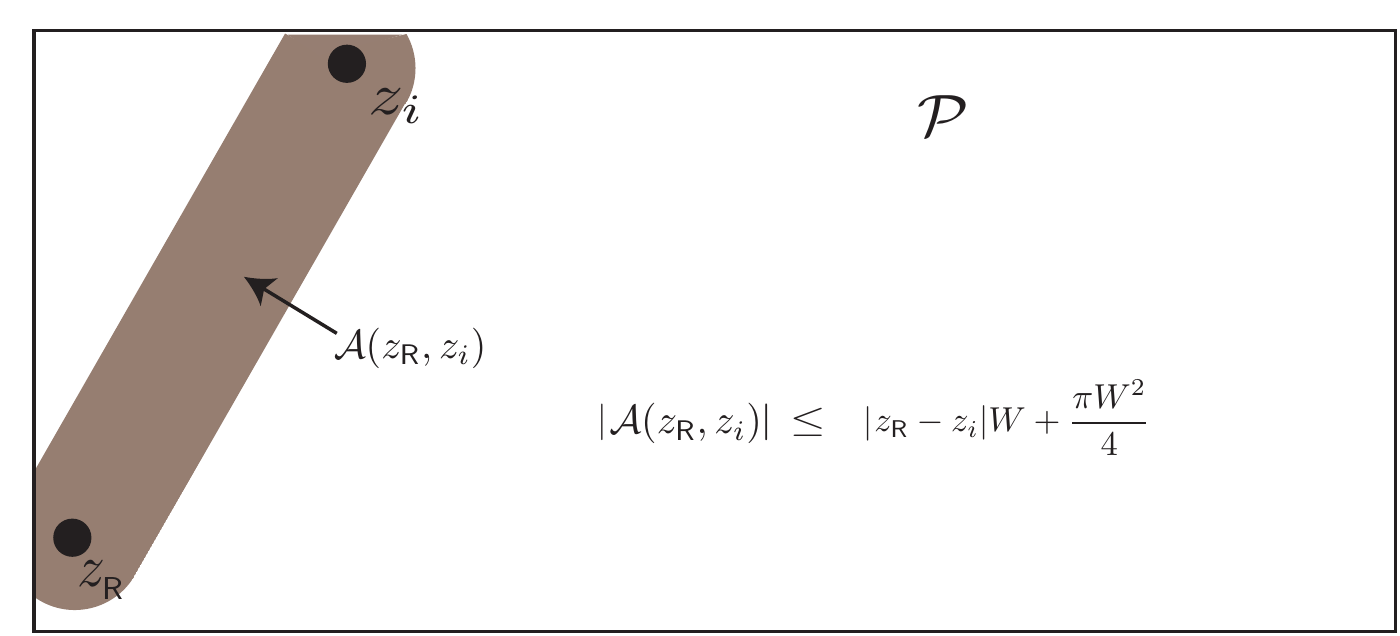}    
\label{fig:blocking_zone_edges}        
}
\caption{Figures showing the blocking zone $\mathcal{A}(\zR,z_i)$, a potential blockage $B_j$, and its blocking cone (cf \cite{mmWave:2015} for definition of blocking cone).}
\end{figure}

\textbf{Assumption 1:} The actual area of $\mathcal{A}'(\zR,z_i)$ can be approximated by the area seen by receiver-transmitter pair positioned away from the reflecting walls so that
\be
|\mathcal{A}'(\zR,z_i)| &\approx& |\zR-z_i|W + \frac{\pi W^2}{4}.
\ee 
\begin{figure}
\centering
\subfigure[The predominant 1$^{st}$ and 2$^{nd}$ order reflections when $z_i$ and $\zR$ are near the walls. Only those reflected paths whose path lengths are close to the direct path length are shown.]{
\includegraphics[totalheight=1.6in,width=3.3in]{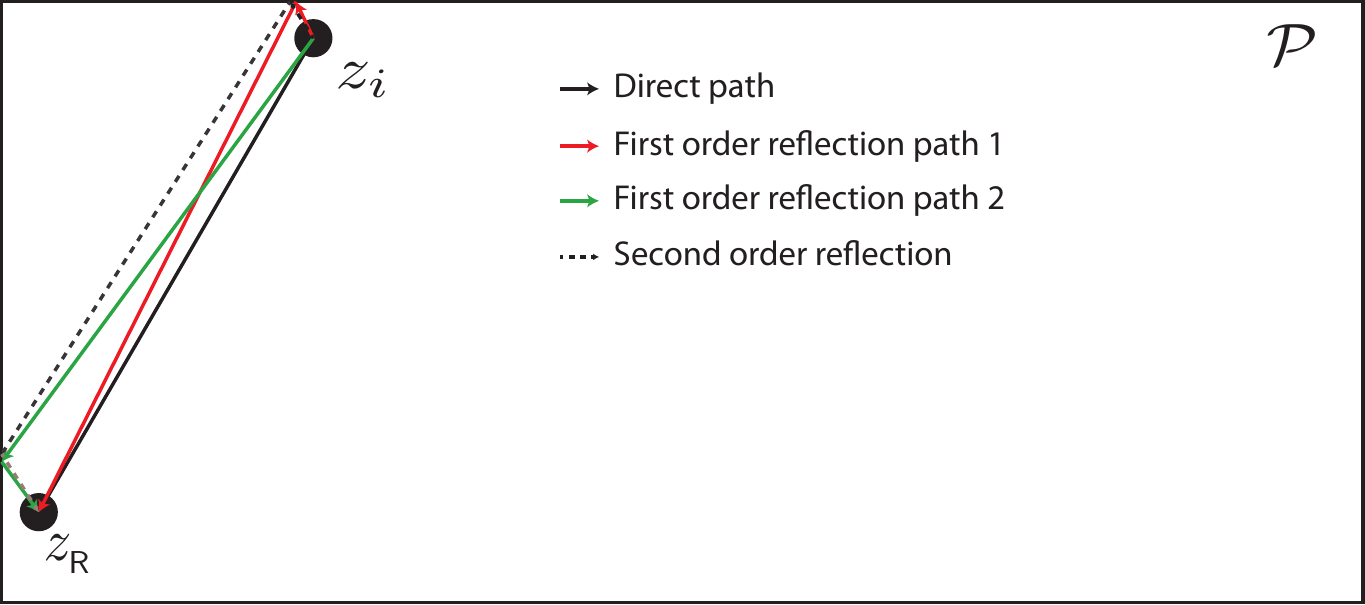}    
\label{fig:wall_reflections}        
}
\subfigure[The actual blocking zone $\mathcal{A}'(\zR,z_i)$ when $z_i$ and $\zR$ are near the reflecting walls.]{
\includegraphics[totalheight=1.6in,width=3.3in]{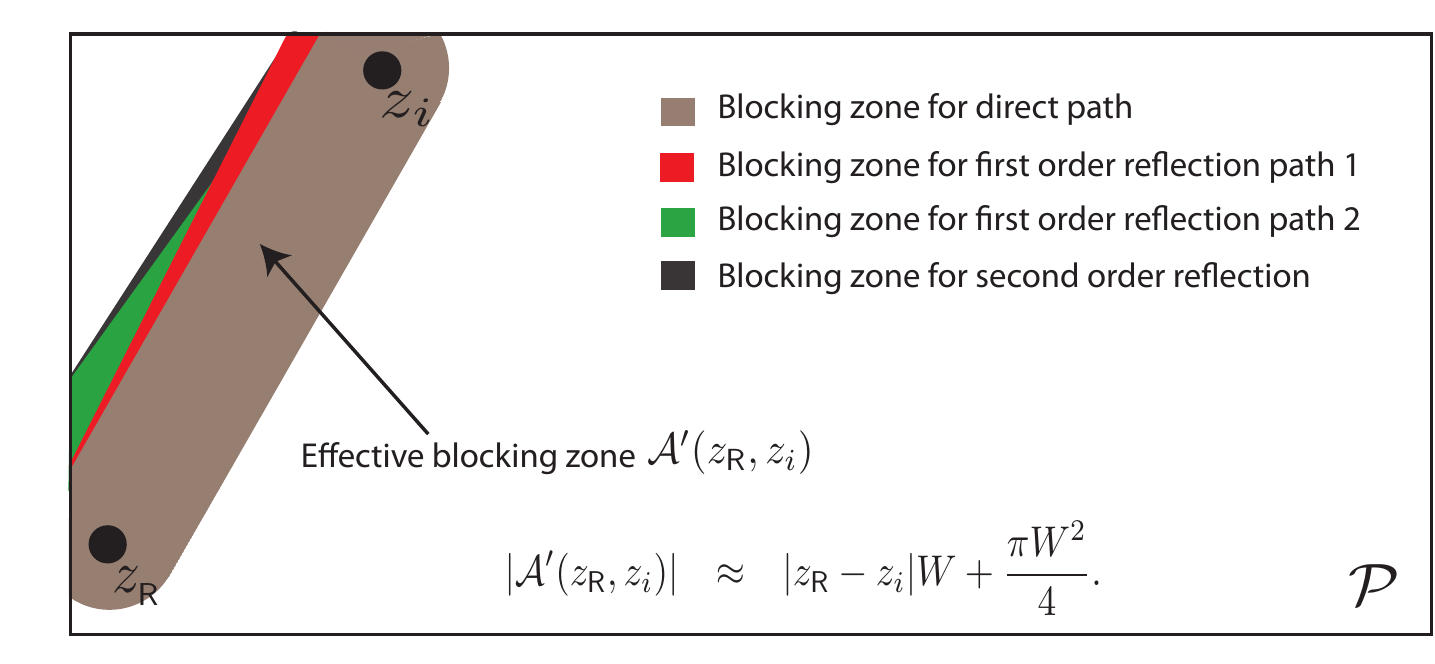}    
\label{fig:blocking_zone_actual}        
}
\caption{Figures showing the blocking zone $\mathcal{A}(\zR,z_i)$, a potential blockage $B_j$, and its blocking cone.}
\end{figure}
With this assumption, the blockage probability $\pb \left( \zR, z_i \right)$ of a user $i$ due to user $j \neq i$ is a function of only the separation between $\zR$ and $z_i$. This is evaluated as
\be
\pb \left( \zR, z_i \right) = 1-\exp\left(-\lambda \left(|\zR-z_i|W + \frac{\pi W^2}{4}\right)\right).
\ee We next evaluate the threshold distance $\RB(\zR)$ next. This is computed in such a way that the average number of interferers that are not blocked is preserved. The average number of strong interferers $\rho (\zR)$ as seen from $\zR$ is
\be
\rho(\zR) \hspace{-0.1in}&=& \hspace{-0.1in}\lambda \int_{z \in \mathcal{P}} (1-\pb (\zR,z)) \mathsf{d}z
\ee
The mean number of interferers in a disk of radius $\RB(\zR)$ around $\zR$ is $\lambda \pi \RB^2(\zR)$, so that equating the mean number of strong interferers leads to
\be
\RB(\zR) &=& \left[ \frac{\rho(\zR)}{\pi}\right]^{\frac{1}{2}}.
\ee We denote this disk around $\zR$ as $\mathcal{B}\left(\zR, \RB(\zR)\right)$. When $\zR$ is near the boundary of $\mathcal{P}$, parts of $\mathcal{B}\left(\zR, \RB(\zR)\right)$ lie outside $\mathcal{P}$. In such a scenario, given that the impact of reflections from the walls is significant, we continue to assume that $\mathcal{B}\left(\zR, \RB(\zR)\right)$ is a complete disk and allow $z_i$ to lie outside $\mathcal{P}$. This is tantamount to modeling the wall reflections as signals emanating from shadow transmitters located at the reflection image locations
corresponding to the actual strong interferers in $\mathcal{P}$. For the ease of analysis, we further assume that (reflection images and actual) strong interferers are independently and uniformly distributed within $\mathcal{B}\left(\zR, \RB(\zR)\right)$. This is illustrated in Fig \ref{fig:final_abstraction}. 

\begin{figure}
\centering
\includegraphics[totalheight=1.5in,width=3.2in]{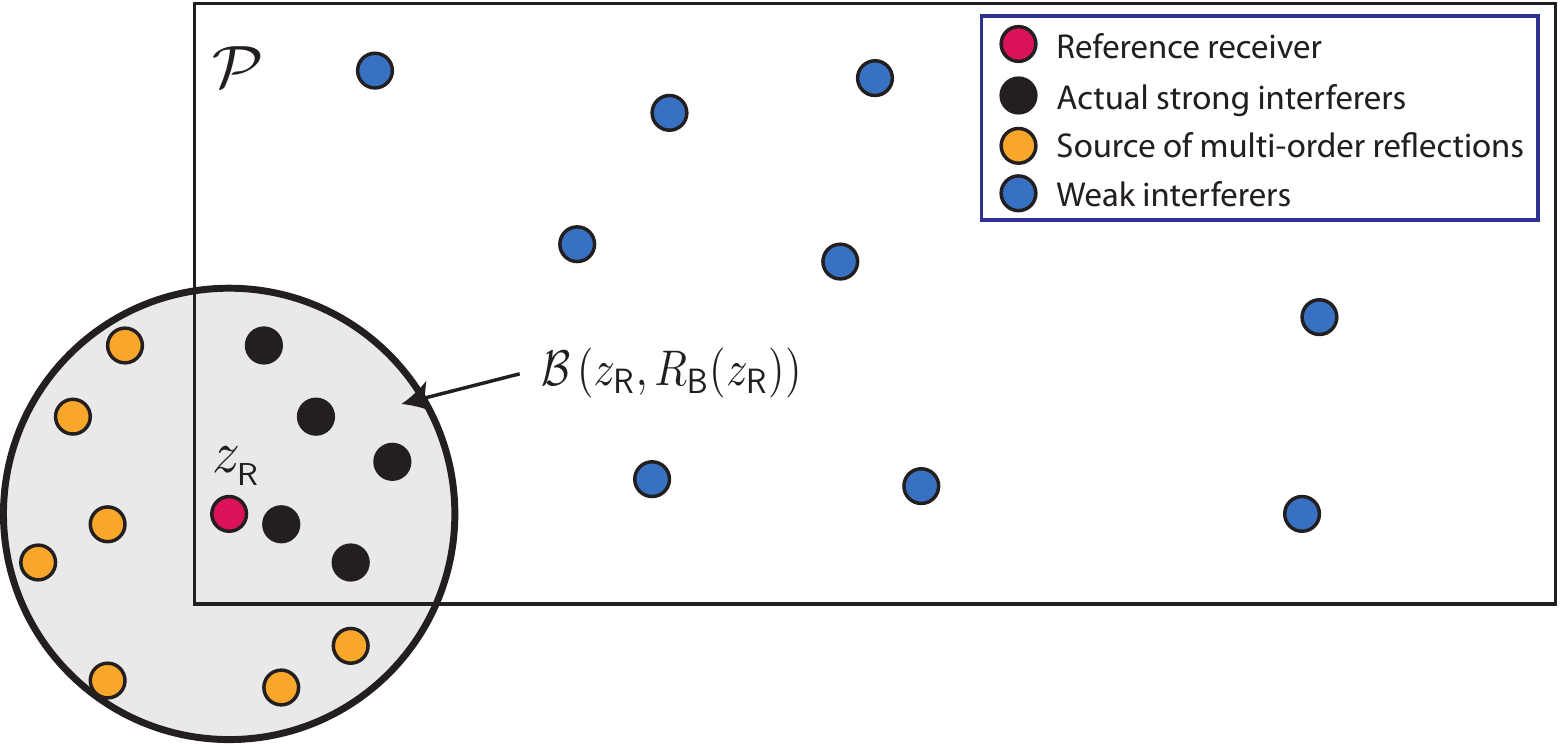}
\caption{Plot showing the region $\mathcal{B}\left(\zR, \RB(\zR)\right)$, when the reference user is near a reflecting wall.}     
\label{fig:final_abstraction}        
\end{figure}

\subsubsection{Self body-blockage}
Self-blockage of the $z_i-\zR$ link occurs if $z_i \in \mathcal{BC}(B_i)$ and/or $z_i \in\mathcal{BC}(\BR)$. This results in a constant attenuation of $\BL$ (linear scale), per number of self-blockages. Such a model has been used in \cite{Bai:Asilomar14} in the context of mmWave cellular systems. In \cite{Bai:Asilomar14}, self-blockage accounts for roughly 40 dB loss in SINR. Unlike the cellular case where the number of self-blockages in a link can be 0 or 1, in the mmWave wearables context, the number of self-blockages in a link can be either 0, 1 or 2. Self-blockage is particularly predominant for strong interferers as signals are otherwise never blocked. For a weak interferer, self-blockage further degrades the signal strength in addition to blockages due to other users. Hence, we assume the propagation is NLOS with a path-loss exponent $\alphaN > \alphaL$. If a weak interferer is not self-blocked, we assume the propagation path via ceiling reflection prevents the channel from being NLOS. 

We summarize our blockage based path-loss model for a general $z_i \in \mathcal{P}$ next. Denoting the number of self-blockages in the $z_i-\zR$ link as $s$ and using ${\mathbf{1}}_{A}$ to denote the indicator function of event $A$, the path-loss function $\ell(\zR,z_i)$ can be classified into any one of the following,

\begin{itemize}
\setlength{\itemindent}{0.25in}
\item[Case A:] When $z_i$ is a strong interferer,
$\ell(\zR, z_i) = {|\zR - z_i|^{-\alphaL}}{\BL^{-s}}.$
\item[Case B:] When $z_i$ is a weak interferer,
$ \ell(\zR, z_i) = \hspace{-0.1in}\left(|\zR - z_i|^2 + (2\ddh)^2\right)^{-\frac{\alphaL}{2}} {\mathbf{1}}_{\{s=0\}}+|\zR - z_i|^{-\alphaN}{\mathbf{1}}_{\{s \neq 0\}}.$
\end{itemize}
Note that we do not consider the reflection from ceiling in Case A. This is because when $z_i$ is a strong interferer located within close proximity to $\zR$ in $\mathcal{P}$, the signal bouncing off the ceiling is less significant in comparison to the direct and wall-reflected signals. The effect of reflections from the ceiling is assumed to be substantial only when the users are facing each other and when $|\zR-z_i| > \RB(\zR)$. The NLOS propagation in Case B-2 coarsely also accounts for all possible scattering and diffraction that dominates when an interferer is weak and self-blocked. This is the intuition behind the different cases for the path-loss function.

\section{SINR Coverage Probability}
\label{Sec:Cov_prob}
The SINR seen at the receiver $\zR$ when its body is facing in the direction $\psi_{\mathsf{R}}$,
\be
\Gamma(\zR,\psi_{\mathsf{R}}) = \frac{h_0 d_0^{-\alphaL}}{\sigma^2 + \sum_{i \in \BPhi} h_i \ell(\zR,z_i)}.
\ee The complementary cumulative distribution function (CCDF) of SINR, which is also referred to as the SINR coverage probability \cite{bai:2014}, is evaluated as
\begin{equation}
\nonumber \mathbb{P}\left( \Gamma(\zR,\psi_{\mathsf{R}}) > \gamma \right)= \mathbb{P}\hspace{-0.04in}\left(\hspace{-0.04in} h_0 > d_0^{\alphaL} \gamma\hspace{-0.04in}\left(\hspace{-0.04in} \sigma^2 + \sum_{i \in \BPhi}h_i \ell(\zR-z_i)\right)\hspace{-0.04in}\right) 
\end{equation}
\begin{equation}
\leq 1-\mathbb{E}_{\BPhi}\left[ \left(1 - e^{-m \tilde{m} \tilde{\gamma} \left(\sigma^2 + \sum_{i \in \BPhi}h_i \ell(\zR,z_i)\right)}\right)^m\right]. \label{Equation:GammaApprox}
\end{equation} In \eqref{Equation:GammaApprox}, we have used a tight lower bound for the CDF of normalized gamma random variable \cite{Alzer:1997}, with $\tilde{m} = {(m!)^{\frac{-1}{m}}}$, and $\tilde{\gamma} = \gamma d_0^{\alphaL}$. Denoting
\be
\ISPhi &=& \sum_{i \in \mathcal{B}\left(\zR, \RB(\zR)\right)}\hspace{-0.2in} h_i  \ell(\zR,z_i)~\mathrm{and} \\
\IWPhi &=& \sum_{i \in \BPhi \setminus \mathcal{B}\left(\zR, \RB(\zR)\right)} \hspace{-0.2in} h_i \ell(\zR,z_i), 
\ee and using the binomial expansion followed by splitting the strong and weak interference terms, we can write \eqref{Equation:GammaApprox} as
\be
\nonumber\mathbb{P}\left( \Gamma(\zR,\psi_{\mathsf{R}}) > \gamma \right) \hspace{-0.1in}&=&\hspace{-0.13in} \sum_{k = 1}^{m}\binom {m}{k} (-1)^{k + 1}e^{-k m \tilde{m} \tilde{\gamma}\sigma^2} \times \\
\hspace{-0.1in}&&\hspace{-0.1in}\mathbb{E}_{\BPhi}\hspace{-0.04in}\left[ e^{-k m \tilde{m} \tilde{\gamma}\ISPhi} \right] \hspace{-0.04in} \mathbb{E}_{\BPhi}\hspace{-0.04in}\left[ e^{-k m \tilde{m} \tilde{\gamma}\IWPhi} \right]\hspace{-0.05in}. \label{Equation:coverage_binom}
\ee Expectation terms in \eqref{Equation:coverage_binom} are as given in Theorem \ref{Theorem} , which makes use of the following lemma.
\begin{lemma}
\label{Lemma}
The probability $p_s$ that $z_i \in \mathcal{B}\left(\zR, \RB(\zR)\right)$ experiences $s$ human body (self) blockages is given by
\be
p_s &=& \begin{cases}
        (1-\pb^{\mathrm{self}})^2 & s=0 \\
        2\pb^{\mathrm{self}}(1-\pb^{\mathrm{self}}) & s=1 \\
        \left(\pb^{\mathrm{self}} \right)^2 & s=2
   \end{cases}, \label{Equation:ps}
\ee where $\pb^{\mathrm{self}} = \frac{1}{\pi}\sin^{-1}\frac{W}{2d}$, and the probability $q(\zR,\psi_{\mathsf{R}})$ that both $z_i$ and $\zR$ are facing each other when $z_i \in \mathcal{P}\setminus \mathcal{B}\left(\zR, \RB(\zR)\right)$ is given by
\be
q(\zR,\psi_{\mathsf{R}}) &=& \left(1-\pb^{\mathrm{self}}\right)\left(1-q_1(\zR,\psi_{\mathsf{R}})\right),
\ee where 
\be
q_1(\zR,\psi_{\mathsf{R}}) &=& \frac{|{\mathcal{BC}(\BR)} \setminus \mathcal{B}\left(\zR, \RB(\zR)\right)|}{LB - |\mathcal{P} \cap \mathcal{B}\left(\zR, \RB(\zR)\right)|}. \label{Equation:q1}
\ee
\end{lemma} 

\begin{proof} 
Since the strong interferers are assumed to be independently and uniformly distributed in $\mathcal{B}\left(\zR, \RB(\zR)\right)$, the probability that $B_i$ blocks $z_i$ is $\pb^{\mathrm{self}}$. The probability that $z_i$ falls in the blocking cone of $\BR$ is also $\pb^{\mathrm{self}}$ as $\mathcal{B}\left(\zR, \RB(\zR)\right)$ is circular with $\zR$ located at its center. So, the probability that both $\BR$ and $B_i$ blocks $z_i$'s interference is $\left(\pb^{\mathrm{self}} \right)^2$, and the probability that neither user bodies block the interference is $(1-\pb^{\mathrm{self}})^2$. Finally, $p_1$ can be computed to satisfy $\sum_{s=0}^{2} p_s = 1$.

The evaluation of $q(\zR,\psi_{\mathsf{R}})$ in the second part of the Lemma is in similar spirit. The probability that a weak interferer is self-blocked due to its own user body is $\pb^{\mathrm{self}}$. The probability of self-blockage due to the reference user's body depends on whether or not the weak interferer lies in the blocking cone of $B_0$ in the region $\mathcal{P} \setminus \mathcal{B}\left(\zR, \RB(\zR)\right)$. As the weak interferers form a PPP of intensity $\lambda(LB - |\mathcal{P} \cap \mathcal{B}\left(\zR, \RB(\zR)\right)|)$, the probability of self-blockage of the weak interferers due to $B_0$ is given by $q_1(\zR,\psi_{\mathsf{R}})$ in \eqref{Equation:q1}. The region outside $\mathcal{B}\left(\zR, \RB(\zR)\right)$ in $\mathcal{P}$ being non-isotropic, this probability needs to be computed numerically for a given $\zR$ and $\psi_{\mathsf{R}}$. 
\end{proof}

\begin{theorem}
\label{Theorem} Denoting $\tilde{R} = \RB(\zR)$, and the region $\mathcal{P}\setminus\mathcal{B}\left(\zR, \RB(\zR)\right)$ as $\mathcal{Q}$ for simplicity,
\be
\mathbb{E}_{\BPhi}\hspace{-0.05in}\left[ e^{-k m \tilde{m} \tilde{\gamma}\ISPhi} \right]\hspace{-0.1in} &=&\hspace{-0.1in} e^{ -2\pi \lambda (\frac{{\tilde{R}}^2}{2}-\overset{2}{\underset{s = 0}{\sum}}p_s \overset{\tilde{R}}{\underset{0}{\int}}(1-(1+ \frac{k  \tilde{m}\tilde{\gamma}}{r^{\alphaL}\BL^{s}})^{-m}r \mathsf{d}r)} \label{Equation:meanSI_integral}\\
\mathbb{E}_{\BPhi}\hspace{-0.05in}\left[ e^{-k m \tilde{m} \tilde{\gamma}\IWPhi} \right] \hspace{-0.1in} &=& \hspace{-0.1in} e^{ -\lambda \left(q(\zR,\psi_{\mathsf{R}}) A_1 + (1-q(\zR,\psi_{\mathsf{R}})) A_2\right)}, \label{Equation:mean_NLOS}
\ee with \vspace{-0.37in}
\be
A_1 \hspace{-0.1in}&=& \hspace{-0.1in}|\mathcal{Q}| -\hspace{-0.1in}\underset{z \in \mathcal{Q}}{\int}\hspace{-0.05in}\left( 1 + \frac{k\tilde{m}\tilde{\gamma}}{({{|\zR-z |}^2+(2\ddh)^2})^{\frac{\alphaL}{2}}}\right)^{-m}\hspace{-0.2in} \mathsf{d}z, \label{Equation:A1}\\ 
\& ~A_2 \hspace{-0.1in} &=& \hspace{-0.1in}|\mathcal{Q}| -\hspace{-0.1in} \underset{z \in \mathcal{Q}}{\int}\left( 1 + \frac{k\tilde{m}\tilde{\gamma}}{{|\zR-z |}^{\alphaN}}\right)^{-m}\hspace{-0.1in} \mathsf{d}z. \label{Equation:A2}
\ee
\end{theorem} 

\begin{proof}
Here we show the proof of the second part of the Theorem, and the first part can be derived along similar lines. Suppose there are $K$ number of weak interferers. Clearly, $K$ is Poisson distributed with mean $\lambda|\mathcal{L}|$, where $|\mathcal{L}| = LB - |\mathcal{P} \cap \mathcal{B}\left(\zR, \RB(\zR)\right)|$ is the area of the region in $\mathcal{P}$ outside the strong interferer ball. With $\ell = \left\lbrace \ell(\zR,z_i) \right\rbrace_{i=1}^{K}$ and $h = \left\lbrace h_i \right\rbrace_{i=1}^{K}$,
\be
\mathbb{E}_{\BPhi}\left[ e^{-k m \tilde{m} \tilde{\gamma}\IWPhi} \right] \hspace{-0.1in}&=& \hspace{-0.1in}\mathbb{E}_{K}\hspace{-0.05in}\left[\mathbb{E}_{\ell,h}\hspace{-0.05in}\left[\prod_{i=1}^{K}e^{-k m \tilde{m} \tilde{\gamma} h_i \ell(\zR,z_i)}\hspace{-0.03in}\right]\hspace{-0.03in}\right]. \label{Equation:Th_e1}
\ee Since $h$ are independent normalized gamma random variables, their moment generating functions can be used to expand \eqref{Equation:Th_e1} as 
\be
\hspace{-0.2in}\mathbb{E}_{\BPhi}\hspace{-0.05in}\left[ e^{-k m \tilde{m} \tilde{\gamma}\IWPhi} \right] \hspace{-0.13in}&=& \hspace{-0.13in}\mathbb{E}_{K}\hspace{-0.05in}\left[\mathbb{E}_{\ell}\hspace{-0.05in}\left[\prod_{i=1}^{K} \left( 1 + k\tilde{m}\tilde{\gamma}\ell(\zR,z_i)\right)^{-m}\right]\hspace{-0.03in}\right]. \label{Equation:Th_e2}
\ee Given $K$, the number of weak interferers falling in the blocking cone of $\BR$ is binomial distributed with parameter $q_1(\zR,\psi_{\mathsf{R}})$ given in \eqref{Equation:q1}. Further, the probability that the interferer is facing $\zR$ is $(1-p_b^{\mathrm{self}})$. Therefore, we can write the RHS of \eqref{Equation:Th_e2} as
\be
\nonumber \hspace{-0.1in}&&\hspace{-0.15in}\mathbb{E}_{K}\hspace{-0.05in}\left[\mathbb{E}_{\{z_i\}_{i=1}^K}\hspace{-0.05in}\left[q(\zR,\psi_{\mathsf{R}}) \left( 1 + \frac{k\tilde{m}\tilde{\gamma}}{({|\zR-z_i |}^2+(2\ddh)^2)^{\frac{\alphaL}{2}}}\right)^{-m} \right. \right .\\
\hspace{-0.1in}&&\hspace{-0.15in}~~+ \left. \left. (1-q(\zR,\psi_{\mathsf{R}}))( 1 + \frac{k\tilde{m}\tilde{\gamma}}{|\zR-z_i |^{\alphaN}})^{-m}\right]^K\right]. \label{Equation:Th_e3}
\ee Finally, since $\{z_i\}_{i=1}^K$ are independent and uniformly distributed in the region $\mathcal{Q}$ (as defined in Theorem \ref{Theorem}), \eqref{Equation:Th_e3} evaluates to the form given in \eqref{Equation:mean_NLOS} when averaged over $K$.
\end{proof}
Plugging the quantities in \eqref{Equation:meanSI_integral} and \eqref{Equation:mean_NLOS} into \eqref{Equation:coverage_binom}, the spatially averaged SINR coverage probability can be computed as function of $\zR$ and $\psi_0$. The spectral efficiency ${\mathcal{C}}(\zR,\psi_{\mathsf{R}})$ for a given SINR can be computed as $\log_2 \left(1 + \Gamma(\zR,\psi_{\mathsf{R}})\right)$. With the knowledge of the CCDF of SINR, the ergodic spectral efficiency $\mathbb{E}\left[ {\mathcal{C}}(\zR,\psi_{\mathsf{R}})\right]$ of the reference user's communication link can be evaluated as a function of the reference receiver location and orientation of its body.
\section{Simulation Results}
\label{Section:sim_results}
In this section, simulation and numerical results that shed insights into the proposed model are discussed. The parameters used for the results are summarized in Table \ref{Table:params}. 

To account for reflections in the simulation, phantom transmitters and user bodies are assumed to be located at the mirror image locations of the actual interferers and their corresponding user body. Then the approach in \cite{mmWave:2015} is used to determine if an interferer is blocked or not. Next, based on the event of self-blockage, the appropriate path-loss model is used. The CCDF of SINR is then obtained by averaging the result for several network realization with users drawn from a density $\lambda$ PPP within the enclosure.
\begin{table}
\caption{Default values of parameters used for simulation}
\label{Table:params}
\centering
\begin{tabular}{|c|c|c|}
\hline 
Parameter & Value & Description\\ 
\hline 
$L$ & 15 m & Length of the enclosure\\ 
\hline 
$B$ & 5 m & Breadth of the enclosure\\ 
\hline 
$H$ & 2.5 m & Height of the enclosure\\ 
\hline 
$\ddh$ & 1 m & Distance of devices from the ceiling\\ 
\hline 
$W$ & 0.45 m & Width of the human-body blockages\\ 
\hline 
$d$ & 0.325 m & Distance of wearable from the user\\ 
\hline
$\lambda$ & 1 m$^{-2}$ & Density of the human users\\ 
\hline
$d_0$ & 0.25 m & Length of the reference link\\ 
\hline
$\alphaL$ & 2 & Path-loss exponent of LOS link\\ 
\hline
$\alphaN$ & 4 & Path-loss exponent of the NLOS link\\ 
\hline
$m$ & 7 & Nakagami parameter for fading\\ 
\hline
$\BL$ & 40 dB & Attenuation due to self-body blockage\\ 
\hline
\end{tabular}
\end{table}
%
The expression derived in Section \ref{Sec:Cov_prob}, is validated against simulation and is shown in Fig. \ref{fig:SINRCCDF} for different values of $\lambda$. Here we assume the reference receiver is facing right, i.e., $\psi_{\mathsf{R}} = 0^o$. The analytic upper bound and the simulation results match. Further, Fig. \ref{fig:SINRCCDF} compares the performance for user densities of $0.5 \text{m}^{-2}$, $2 \text{m}^{-2}$, and $4 \text{m}^{-2}$. With higher user density, the SINR coverage probability reduces since the system becomes more interference limited. 
\begin{figure}
\centering
\includegraphics[totalheight=3.2in,width=3.2in]{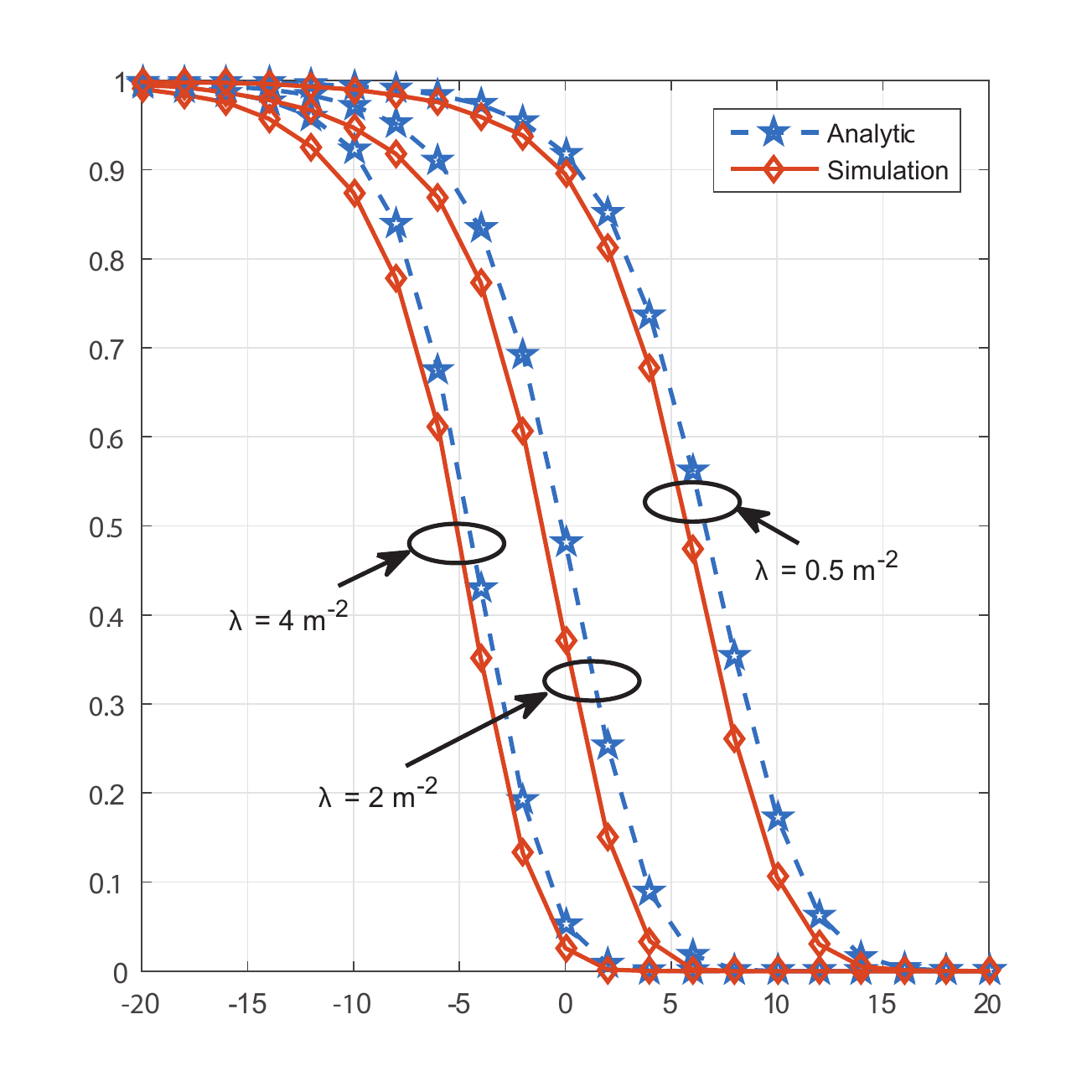}
\caption{SINR distribution obtained through simulation and analytic expression when the receiver is at the center for different values of user density $\lambda$.}     
\label{fig:SINRCCDF}        
\end{figure}

The dependence of system performance on the reference user body orientation is studied by plotting the average achievable rate for the cases when the reference user is at the center and near a corner (we assume $\zR = 0.5 + \jj 0.5$ for this case). This is computed by multiplying the ergodic spectral efficiency with the system bandwidth which is taken to be $1.76 ~\mathrm{GHz}$ assuming an IEEE 802.11ad like single carrier PHY setup. The plots are shown in Fig. \ref{fig:rate_orientation} from which we see that the sensitivity to body orientation is more pronounced when the reference user is at corner. Moreover, the system performance is better when the reference user is at a corner since interference is less pronounced that when at the center of the enclosed space.

\begin{figure}
\vspace{-0.1in}
\centering
\includegraphics[totalheight=3.2in,width=3.2in]{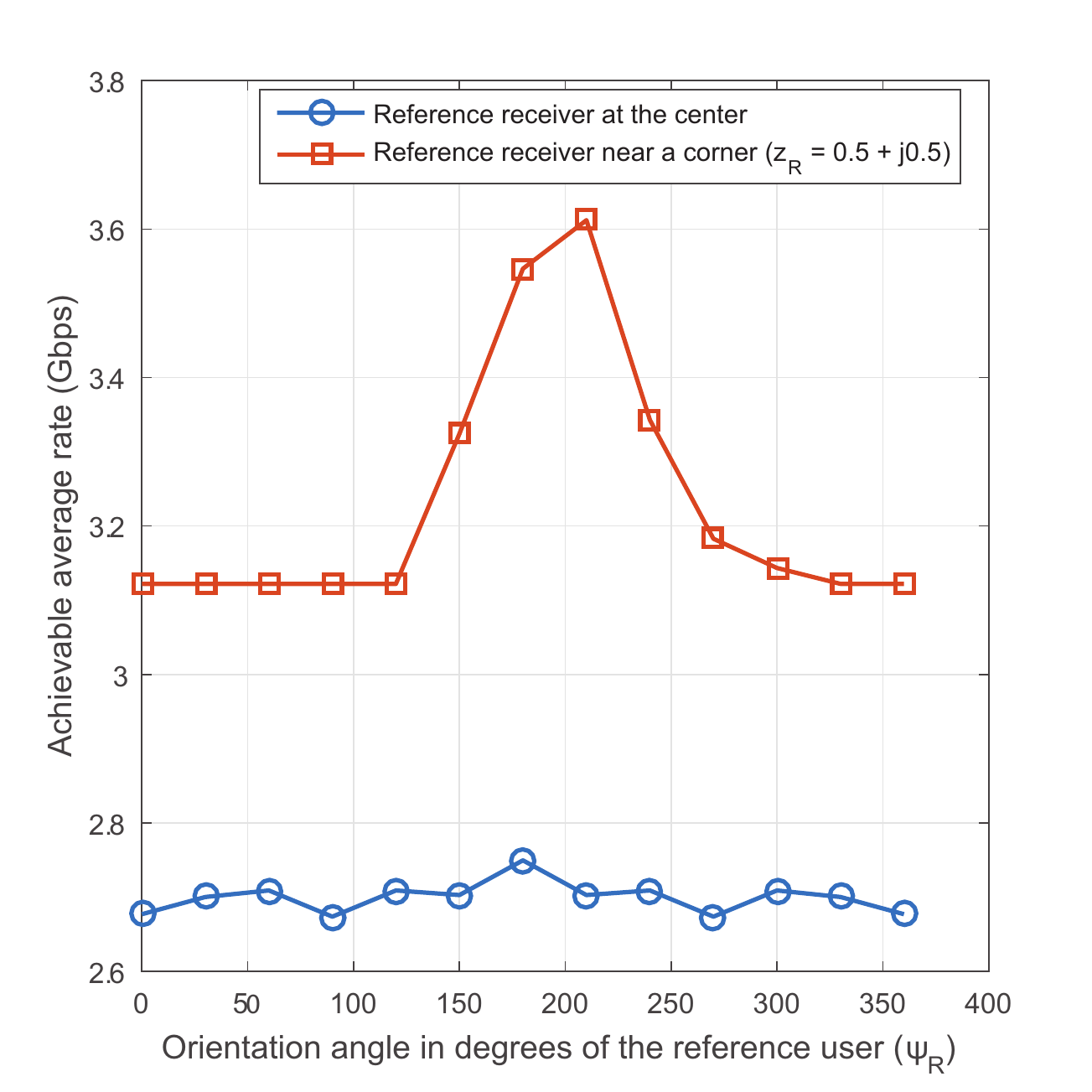}
\caption{Plot showing the variation in the average achievable rate as a function of the body orientation of the reference user when located at the center and near a corner.}     
\label{fig:rate_orientation}        
\end{figure}

The location dependent SINR coverage probability is shown in Fig. \ref{fig:heat_map} as a heat-map for an SINR threshold of $3~\mathrm{dB}$ and $\psi_{\mathsf{R}} = 180^o$. For this case, the best performance is obtained when the reference receiver is near the left wall and facing away from the interfering crowd. This is the case when all the interferers are shielded by the reference user's body.

\section{Conclusion}
In this paper, we proposed a tractable system model to capture the effects of body blockages, wall and ceiling reflections in a dense indoor mmWave wearable network. Closed form expressions were derived as a function of the location and body orientation of a reference user. The proposed model enables us to evaluate spatially averaged system performance without the need to conduct elaborate simulations. The key parameters involved in the model are a threshold distance from the reference receiver where the interferers are strong, and a probabilistic characterization of the self-body blockage. While the threshold distance is a function of the reference user's location, user density and the dimensions of the enclosure, self-blockage probability is a function of the users' relative body orientation.

It was observed that the effect of body orientation is significant when the reference user is located at a corner. The peak average rate for the system was obtained when the reference user is near the corner and facing away from the interferers. Further gains can be achieved by the use of directional antennas that results in directed transmission and reception. Using multiple antennas, the elevation angles of the antenna main-lobe can be leveraged to minimize interference. This study is saved for future work. It would also be interesting to see if tractable models can be derived when different wearable devices are located at different heights.

\begin{figure}[t]
\vspace{-0.8in}
\centering
\includegraphics[totalheight=3.2in,width=3.2in]{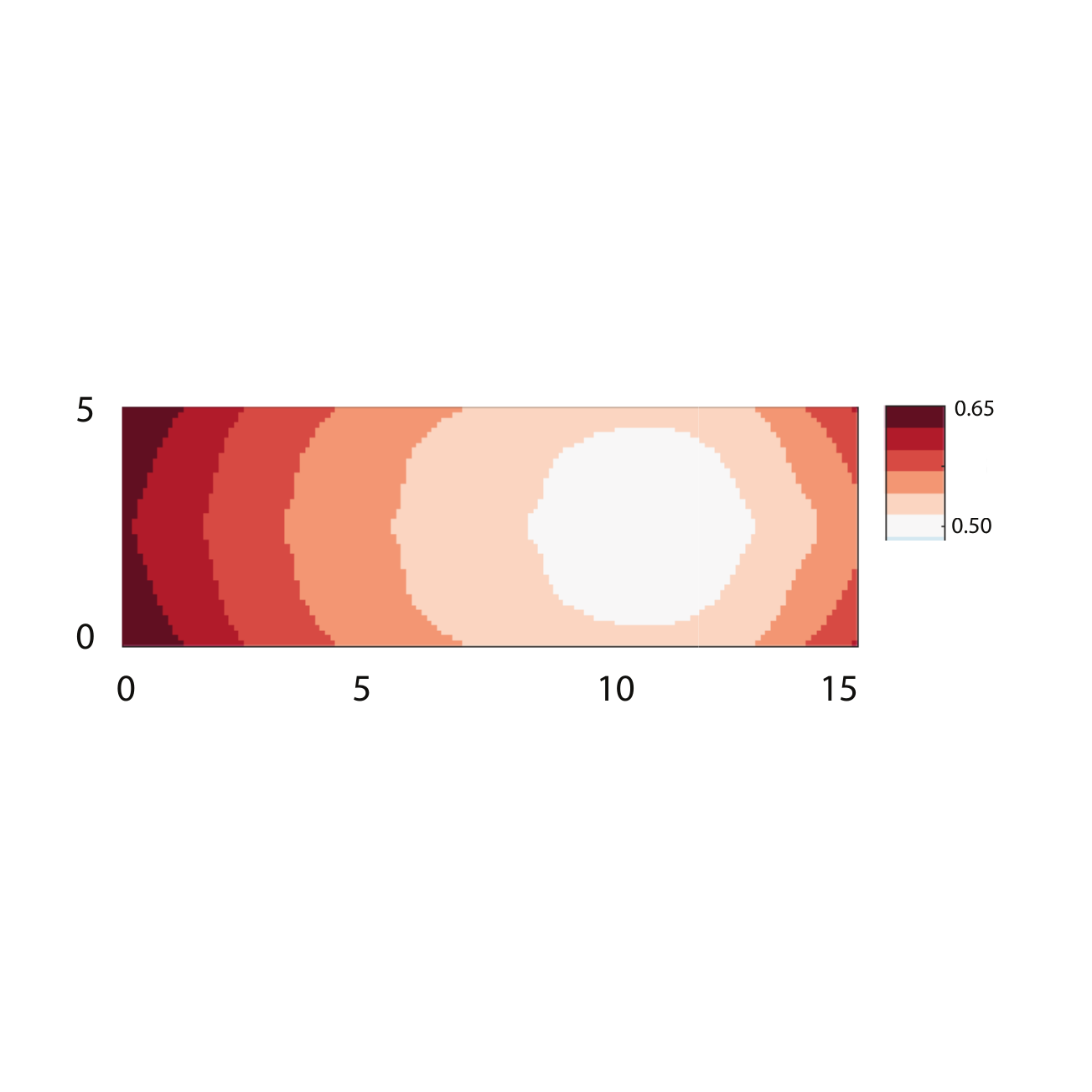}
\vspace{-1.2in}
\caption{SINR coverage probability heat-map as a function of the reference location position when the reference user is facing to the left, i.e. $\psi_{\mathsf{R}} = 180^o$.}     
\label{fig:heat_map}        
\end{figure}

\bibliographystyle{ieeetr}

\end{document}